\documentclass{lmcs}

\usepackage{enumerate}
\usepackage{hyperref}
\usepackage{amsmath, amssymb}

\theoremstyle{theorem}
\newtheorem{lemma}{Lemma}[section]
\newtheorem{theorem}[lemma]{Theorem}
\newtheorem{proposition}[lemma]{Proposition}
\newtheorem{corollary}[lemma]{Corollary}
\theoremstyle{definition}

\newtheorem{example}[lemma]{Example}

\newcommand{\In}{\subseteq}

\newcommand{\cl}{\mathop{\rm cl}}
\newcommand{\IN}{{\mathbb N}}

\def\supt{\mbox{$\sup^{(2)}$}}

\begin{document}

\title[]{Two Counterexamples Concerning the Scott Topology on a Partial Order}

\author[]{Peter Hertling}
\address{Fakult\"at f\"ur Informatik, Universit\"at der Bundeswehr M\"unchen, 85577
Neubiberg, Germany}
\email{peter.hertling@unibw.de}

\thanks{{\em 2012 ACM Computing Classification System:}
Theory of Computation - Logic,
Mathematics of computing - Continuous Mathematics - Topology}

\keywords{
Directed complete partial order;
supremum function;
product topology;
continuity;
bounded completeness}

\thanks{The author was supported by the EU grant FP7-PEOPLE-2011-IRSES No. 294962: COMPUTAL.}


\begin{abstract}
We construct a complete lattice $Z$ such that the binary supremum
function $\sup:Z\times Z\to Z$ is discontinuous with respect to the product topology on $Z\times Z$
of the Scott topologies on each copy of $Z$. In addition, we show that bounded completeness
of a complete lattice $Z$ is in general not inherited by the dcpo $C(X,Z)$ of continuous
functions from $X$ to $Z$ where $X$ may be any topological space and where on $Z$
the Scott topology is considered.
\end{abstract}

\maketitle

\section{Introduction}
\label{section:intro}

In this note two counterexamples in the theory of partial orders are constructed.
The first counterexample concerns the question whether the binary supremum function
on a sup semilattice is continuous with respect to the product topology of the 
Scott topology on each copy of the sup semilattice.
The second counterexample concerns the question whether bounded completeness
of a complete lattice is inherited by the dcpo of continuous functions from some
topological space to the complete lattice.

The Scott topology on a partially ordered set (short: {\em poset}) $(X,\sqsubseteq)$
has turned out to be important in many areas of computer science.
Given two posets $(X_1,\sqsubseteq_1)$ and $(X_2,\sqsubseteq_2)$, one can
in a natural way define a partial order relation $\sqsubseteq_\times$ on the product
$X_1 \times X_2$. 
This gives us two natural topologies on the
product $Z= X_1 \times X_2$:
\begin{itemize}
\item
The product of the Scott topology on $(X_1,\sqsubseteq_1)$ and of the
Scott topology on $(X_2,\sqsubseteq_2)$,
\item
the Scott topology on $(Z,\sqsubseteq_\times)$.
\end{itemize}
It is well known that the second topology on $Z=X_1\times X_2$ is always at
least as fine as the first topology, and that there are examples where the
second topology is strictly finer than the first topology. 
The best known example of this kind seems to
consist of the Johnstone space (compare \cite{Joh81}) for $(X_1,\sqsubseteq)$
and the dcpo of its Scott open subsets for $(X_2,\sqsubseteq)$;
see \cite[Exercise II-4.26 and Exercise II.1-36]{GHKLMS03}
and \cite[Exercise 5.2.16 and Exercise 5.2.15]{Gou13}.

If the supremum of any two elements of a poset $(X,\sqsubseteq)$ exists
(in this case the partial order is called a {\em sup semilattice})
then the function $\supt:X\times X\to X$ is well defined. Is it continuous? It
is easy to see that it is continuous if on $X\times X$ the finer topology of the two
topologies discussed above, the Scott topology, is considered. 
The natural question arises whether
the function $\supt:X \times X\to X$ is still continuous if on $X\times X$
one considers the coarser topology, the product topology.
It is the first main goal of this note to show
that in general this is not the case.
In fact, we are going to construct even a complete lattice $(X,\sqsubseteq)$
such that the function $\supt:X\times X\to X$ is not continuous with respect
to the product topology on $X\times X$.

To the best of my knowledge such an example was not yet known.
This topic was considered in the book 
\cite[Page 139]{GHKLMS03}
and in its predecessor \cite{GHKLMS80} (compare \cite[Page 106]{GHKLMS80}):
``{\em We cannot be satisfied by saying that the sup operation
is a continuous function $(L\times L,\sigma(L\times L)) \to (L,\sigma(L))$; this 
continuity is weaker, since in general we have a proper containment of topologies:
$\sigma(L\times L) \supset \sigma(L) \times \sigma(L)$.
We will return to this question at greater length in Section II-4 below
(see II-4.13 ff.).}''
Note that in this quote $L$ is a poset, and
$\sigma(L)$ denotes the Scott topology on $L$.
In turns out that in \cite[Section II-4]{GHKLMS03} the authors consider conditions
that imply that the product topology on $X\times X$ is identical with the Scott
topology on $X\times X$. Then, of course, 
the function $\supt:X\times X\to X$ is continuous with respect to the
product topology on $X\times X$ as well
(provided that $(X,\sqsubseteq)$ is a sup semilattice).
But it seems that the general question whether the function
$\supt:X\times X\to X$ must be continuous if on $X\times X$ the product
topology is considered has not been answered so far.
In the book \cite{GHKLMS03} I was not able to find an example of a 
sup semilattice $(X,\sqsubseteq)$ such that the function
$\supt:X\times X\to X$ is not continuous with respect to the
product topology on $X$. I did also not find a construction of such an example in any
other source on domain theory that I consulted.

Our second counterexample concerns bounded completeness of the dcpo of 
continuous functions from a topological space to a bounded complete dcpo.
A poset is called {\em bounded complete} if each bounded subset has a supremum.
The following question arises: If $(Z,\sqsubseteq)$ is a bounded complete dcpo
and $X$ a topological space, is then the dcpo $C(X,Z)$ of continuous functions
from $X$ to $Z$ bounded complete as well? It turns out that the answer
to this question is yes if
the topology on $X$ is the Scott topology of some partial order on $X$. 
One may now speculate whether the answer is yes even for arbitrary topological spaces.
For example, Edalat \cite[Page 502]{Eda08}
writes
``{\em For any topological space $Z$ and any bounded complete dcpo $D$ with bottom
$\bot$, let $Z \to  D$ be the bounded complete dcpo of Scott continuous functions from
$Z$ to $D$.}''
But it is the second main goal of this note to show that in general the answer to this
question is no.
In Section~\ref{section:bc}
we shall even construct a topological space $X$ and a complete lattice
$Z$ such that the dcpo of continuous functions from $X$ to $Z$ is not bounded complete.
This construction makes essential use of the first counterexample (in Section~\ref{section:sup}).
Note that in \cite[Page 26]{KL14} the closely related question is discussed
whether the pointwise supremum of two Scott continuous
functions $f,g:X \to Z$ for a topological space $X$ and a complete lattice $Z$
is again Scott continuous. The authors observe that this is true if the product topology
on $Z\times Z$ is identical with the Scott topology on $Z\times Z$.
Concerning the general case, they write ``{\em But in general, we cannot
conclude that $f \vee g$ or $f \wedge g$ are Scott continuous}''.
But they do not give an example where the pointwise
supremum $f\vee g$ of $f$ and $g$ is not Scott continuous.
In Theorem~\ref{theorem:bounded-complete-counterexample},
the main result of Section~\ref{section:bc}, such an example is constructed.

After collecting several elementary notions concerning partial orders in the following section,
in Section~\ref{section:johnstone} we modify the Johnstone space in order to obtain a somewhat
simpler partial order which is even a complete lattice. 
In Section~\ref{section:sup} we construct a complete lattice $(X,\sqsubseteq)$ such that
such that the function $\supt:X\times X\to X$ is not continuous with respect
to the product topology on $X\times X$.
Finally, in Section ~\ref{section:bc} we
construct a topological space $X$ and a complete lattice
$Z$ such that the dcpo of continuous functions from $X$ to $Z$ is not bounded complete
and such that there exist two Scott continuous functions $f,g:X\to Z$ whose pointwise
supremum is not Scott continuous.

\section{Basic Notions for Partially Ordered Sets}
\label{section:cpo}

In this section, for the convenience of the reader, we collect several elementary notions
concerning partial orders as well as the definition of the Scott topology.

Let $Z$ be a set. A binary relation $\sqsubseteq \In Z \times Z$ satisfying
the following three conditions
\begin{enumerate}
\item
$(\forall z)\ z \sqsubseteq z$ ({\em reflexivity}),
\item
$(\forall x,y,z)\ (x \sqsubseteq y \wedge y \sqsubseteq z) \Rightarrow x \sqsubseteq z$ ({\em transitivity}),
\item
$(\forall y,z)\ (y \sqsubseteq z \wedge z \sqsubseteq y) \Rightarrow y = z$ ({\em antisymmetry}),
\end{enumerate}
is called a {\em partial order}. 
A pair $(Z,\sqsubseteq)$ consisting of a nonempty set $Z$ and a partial order
$\sqsubseteq$ on $Z$ is called a {\em partially ordered set} or {\em poset}. 
Let $(Z,\sqsubseteq)$ be a poset.
\begin{itemize}
\item
An element $z\in Z$ is called an {\em upper bound} of a subset $S \In Z$
if $(\forall s \in S)\ s \sqsubseteq z$.
\item
An element $z\in Z$ is called a {\em supremum} or {\em least upper bound}
of a subset $S \In Z$ if it is an upper bound of $S$ and if for all upper bounds
$y$ of $S$ one has $z \sqsubseteq y$.
Obviously, if a set $S$ has a supremum, then this supremum is unique.
Then we denote it by $\sup(S)$.
\item
An element $z\in Z$ is called a {\em lower bound} of a subset $S \In Z$
if $(\forall s \in S)\ z \sqsubseteq s$.
\item
An element $z\in Z$ is called an {\em infimum} or {\em greatest lower bound}
of a subset $S \In Z$ if it is a lower bound of $S$ and if for all lower bounds
$y$ of $S$ one has $y \sqsubseteq z$.
Obviously, if a set $S$ has an infimum, then this infimum is unique.
Then we denote it by $\inf(S)$.
\item
The poset $(Z,\sqsubseteq)$ is called a {\em sup semilattice}
if $\sup\{x,y\}$ exists for all $x,y\in Z$.
\item
The poset $(Z,\sqsubseteq)$ is called a {\em complete lattice}
if $\sup(S)$ exists for all subsets $S \subseteq Z$.
Note that then also $\inf(S)$ exists for all subsets $S \subseteq Z$.
Indeed, it is well known --- compare, e.g., \cite[Prop. O-2.2(i)]{GHKLMS03} ---
and straightforward to check that
if $\sup(S)$ exists for all subsets $S\subseteq Z$ then for any subset
$T\subseteq Z$, the supremum of the set of all lower bounds of $T$ is
an infimum of $T$.
\item
A subset $S \In Z$ is called {\em upwards closed} if for any elements $s,z\in Z$:
if $s\in S$ and $s \sqsubseteq z$ then $z \in S$.
\item
A subset $S \In Z$ is called {\em directed} if it is nonempty and for any
two elements $x,y \in S$ there exists an upper bound $z \in S$ of the set $\{x,y\}$.
\item
The poset $(Z,\sqsubseteq)$ is called a {\em dcpo} if
for any directed subset $S \In Z$ there exists a supremum of $S$ in $Z$.
\end{itemize}

The following lemma is well known. 

\begin{lemma}[{See, e.g., Goubault-Larrecq~\cite[Proposition~4.2.18]{Gou13}}]
\label{lemma:scott-topology}
Let $(Z,\sqsubseteq)$ be a poset.
The set of all subsets $O \In Z$ satisfying the following two conditions:
\begin{enumerate}
\item
$O$ is upwards closed,
\item
for every directed subset $S\In Z$ such that $\sup(S)$ exists and $\sup(S) \in O$
the intersection $S \cap O$ is not empty,
\end{enumerate}
is a topology on $Z$, called the {\em Scott topology}.
\end{lemma}

\section{A Modification of the Johnstone Space}
\label{section:johnstone}

Let $(X_1,\sqsubseteq_1)$ and $(X_2,\sqsubseteq_2)$ be two posets.
On the product $Z := X_1 \times X_2$ we define a
binary relation $\sqsubseteq_\times \subseteq Z \times Z$ by
\[ (x_1,x_2) \sqsubseteq_\times (x_1',x_2') :\iff
   (x_1 \sqsubseteq_1 x_1' \text{ and } x_2 \sqsubseteq_2 x_2') . \]
   
\begin{lemma}
\label{lemma:prod}
\begin{enumerate}
\item
$(Z,\sqsubseteq_\times)$ is a poset as well.
\item
If $(X_1,\sqsubseteq_1)$ and $(X_2,\sqsubseteq_2)$ are dcpo's then
$(Z,\sqsubseteq_\times)$ is a dcpo as well.
\item
If $(X_1,\sqsubseteq_1)$ and $(X_2,\sqsubseteq_2)$ are complete lattices then
$(Z,\sqsubseteq_\times)$ is a complete lattice as well.
\item
Any subset $U \subseteq Z$ open in the product topology on $Z=X_1 \times X_2$
of the Scott topology on $(X_1,\sqsubseteq_1)$ and the
Scott topology on $(X_2,\sqsubseteq_2)$ is open in the Scott
topology on $(Z,\sqsubseteq_\times)$.
\end{enumerate}
\end{lemma}

As all the assertions in this lemma are well known and easy to check we omit their proofs.

The notions developed here give us two natural topologies on the
product $Z= X_1 \times X_2$:
\begin{itemize}
\item
The product of the Scott topology on $(X_1,\sqsubseteq_1)$ and of the
Scott topology on $(X_2,\sqsubseteq_2)$,
\item
the Scott topology on $(Z,\sqsubseteq_\times)$.
\end{itemize}
According to the fourth assertion of the previous lemma the second topology
is at least as fine as the first topology. The natural question whether the two topologies
must always coincide or not has a negative answer.

\begin{proposition}
\label{prop:johnstone}
There exist dcpo's $(X_1,\sqsubseteq_1)$ and $(X_2,\sqsubseteq_2)$ such that
the Scott topology on $(Z,\sqsubseteq_\times)$ is strictly finer than
the product topology on $Z=X_1 \times X_2$
of the Scott topology on $(X_1,\sqsubseteq_1)$ and the
Scott topology on $(X_2,\sqsubseteq_2)$.
\end{proposition}

This is well known as well, and the standard example seems to be 
$X_1 :=$ the so-called Johnstone space, $X_2 :=$ the dcpo of the Scott open subsets of $X_1$;
compare \cite[Exercise II-4.26 and Exercise II.1-36]{GHKLMS03}
and \cite[Exercise 5.2.16 and Exercise 5.2.15]{Gou13}.
Since in the following we are going to define a similar but slightly simpler partial order
for comparison in the following example we define the Johnstone space~\cite{Joh81}
and sketch the well-known proof of Proposition~\ref{prop:johnstone}.
Let $\IN=\{0,1,2,\ldots\}$ be the set of non-negative integers.

\begin{example}
\label{example:johnstone}
Let $\leq$ on $\IN\cup\{\omega\}$ be the usual linear ordering with $n\leq\omega$
for all $n \in\IN\cup\{\omega\}$.
On $X_1 := \IN \times (\IN\cup\{\omega\})$ we define a binary relation $\sqsubseteq_1$
by
\[ (i,j) \sqsubseteq_1 (k,l) :\iff ((i=k \text{ and } j \leq l) \text{ or }
             (l = \omega \text{ and } j \leq k)) . \]
Then the pair $(X_1,\sqsubseteq_1)$ is a dcpo. It is called the {\em Johnstone space}.
Let $X_2$ be the set of all Scott open subsets of $X_1$,
and let the binary relation $\sqsubseteq_2$ on $X_2$ be set-theoretic inclusion.
Then $(X_2,\sqsubseteq_2)$ is a dcpo as well. One can check that
the subset $E \subseteq X_1 \times X_2$ defined by
\[ E := \{(x_1,x_2) \in X_1 \times X_2 ~:~ x_1 \in x_2\} \]
is open in the Scott topology on $(X_1 \times X_2, \sqsubseteq_\times)$
but not open in the product topology of  the Scott topology on $(X_1,\sqsubseteq_1)$
and the Scott topology on $(X_2,\sqsubseteq_2)$.
\end{example}

For later use, we explicitly formulate one general step in the previous example.

\begin{lemma}
\label{lemma:elementof}
Let $(X_1,\sqsubseteq_1)$ be a poset,
and let $(X_2,\sqsubseteq_2)$
be the poset consisting of the Scott open subsets of $X_1$
where $\sqsubseteq_2$ is set-theoretic inclusion.
Then the subset $E \subseteq X_1 \times X_2$ defined by
\[ E := \{(x_1,x_2) \in X_1 \times X_2 ~:~ x_1 \in x_2\} \]
is open in the Scott topology on $(X_1 \times X_2, \sqsubseteq_\times)$.
\end{lemma}

The proof is straightforward. We omit it.
By a modification (which is mostly a simplification) of the definition of the Johnstone space
we can prove the following slight improvement of Prop.~\ref{prop:johnstone}.
It will be used in the construction of the counterexample in Section~\ref{section:sup}.
Note that the theorem deduced in Section~\ref{section:sup} from the
following proposition is actually stronger than
the following proposition.

\begin{proposition}
\label{proposition:mod-johnstone}
There exist complete lattices $(X_1,\sqsubseteq_1)$ and $(X_2,\sqsubseteq_2)$ such that
the Scott topology on $(X_1 \times X_2,\sqsubseteq_\times)$ is strictly finer than
the product topology on $X_1 \times X_2$
of the Scott topology on $(X_1,\sqsubseteq_1)$ and the
Scott topology on $(X_2,\sqsubseteq_2)$.
\end{proposition}

Our modification of the Johnstone space is defined as follows. Let 
\[ X_1 := \{\bot\} \cup (\IN \times (\IN \cup\{\omega\})) \cup\{\top\} . \]
We define a binary relation $\sqsubseteq_1$ on $X_1$ by
\[ x \sqsubseteq_1 y :\iff \begin{array}[t]{l}
       x = \bot \text{ or } y = \top \text{ or } \\
       (\exists i,j \in \IN)(\exists p,q \in \IN\cup\{\omega\}) ( x=(i,p) \text{ and } y=(j,q) \text{ and } \\
       ((i=j \text{ and } p \leq q) \text{ or } (i \leq j \text{ and } q = \omega))) 
       \end{array} \]
for all $x,y \in X_1$.
Compare Figure~\ref{figure:simpler-than-johnstone}.
\begin{figure}[htbp]
\begin{center}
\includegraphics[height=6.5cm]{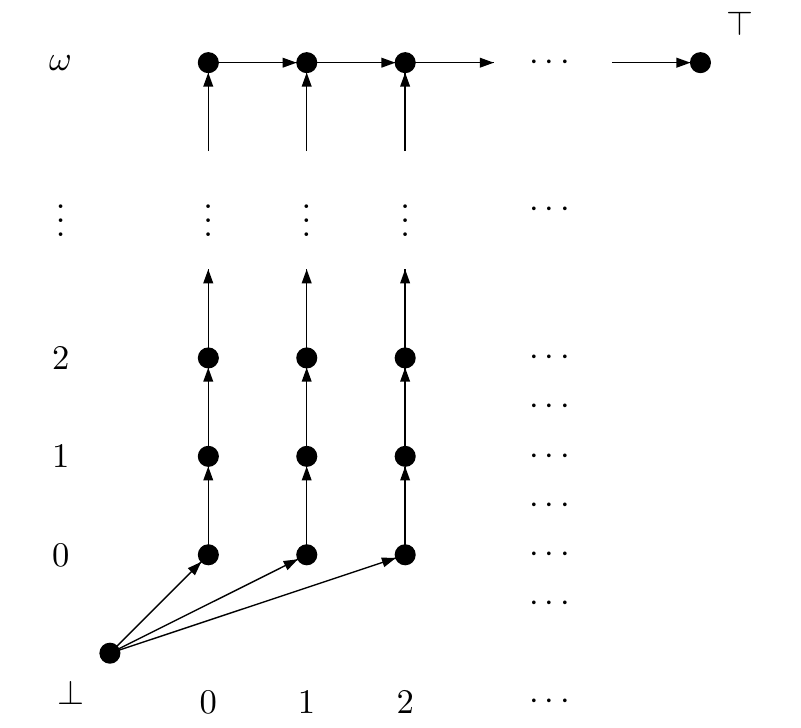}
\caption{The binary relation $\sqsubseteq_1$ on
$X_1 = \{\bot\} \cup (\IN \times (\IN \cup\{\omega\})) \cup\{\top\}$
defined in Section~\ref{section:johnstone}
after Proposition~\ref{proposition:mod-johnstone}
is the reflexive and transitive closure of the binary relation depicted in this figure.}
\label{figure:simpler-than-johnstone}
\end{center}
\end{figure}
It is clear that $(X_1,\sqsubseteq_1)$ is a poset.
By distinguishing four cases of subsets $S \subseteq X_1$,
we observe that $\sup(S)$ is defined for any subset $S\subseteq X_1$:
\begin{enumerate}
\item
If $S\subseteq\{\bot\}$ then $\sup(S)=\bot$.
\item
If $S\not\subseteq\{\bot\}$ but there exists some $i\in\IN$ such that
$S \subseteq \{\bot\} \cup (\{i\} \times (\IN \cup\{\omega\}))$ then
$\sup(S)=(i,\sup\{j \in \IN\cup\{\omega\} ~:~ (i,j) \in S\})$.
\item
If there does not exist an $i\in\IN$ with
$S \subseteq \{\bot\} \cup (\{i\} \times (\IN \cup\{\omega\}))$
but there exists a $k\in\IN$ such that
$S \subseteq \{\bot\} \cup (\{0,1,\ldots,k\} \times (\IN \cup\{\omega\}))$
then
$\sup(S)=(\max\{j \in \IN ~:~ S \cap (\{j\} \times (\IN \cup \{\omega\})) \neq \emptyset\},\omega)$.
\item
If there does not exist a $k\in\IN$ such that
$S \subseteq \{\bot\} \cup (\{0,1,\ldots,k\} \times (\IN \cup\{\omega\}))$
then $\sup(S)=\top$.
\end{enumerate}
Thus, $\sup(S)$ exists for every subset $S\subseteq X_1$.
Thus, $(X_1,\sqsubseteq_1)$ is a complete lattice.

Let us characterize the Scott open subsets of $(X_1,\sqsubseteq_1)$.
For a number $n\in\IN$ we define
\[ \IN_n:=\{m \in \IN ~:~ m \geq n\} . \]
For a number $n\in\IN$ and a function $f:\IN_n\to \IN$ we define
a subset $V_f \subseteq X_1$ as follows:
\[ V_f:= \{ (i,j) \in \IN_n\times (\IN\cup\{\omega\}) ~:~ f(i) \leq j\}
   \cup\{\top\} . \]
   
\begin{lemma}
\label{lemma:mod-johnstone-1}
\begin{enumerate}
\item
For every $n\in\IN$ and every function $f:\IN_n\to\IN$ the set $V_f$
is a Scott open subset of $(X_1,\sqsubseteq_1)$.
\item
For every Scott open subset $U\subseteq X_1$
that is neither the empty set $\emptyset$ nor the full set $X_1$
there exist some $n\in\IN$ and some function $f:\IN_n\to\IN$ with $U=V_f$. 
\end{enumerate}
\end{lemma}

\begin{proof}
\begin{enumerate}
\item
Consider some $n\in\IN$ and some function $f:\IN_n\to\IN$.
It is clear that $V_f$ is upwards closed.
By going through the four cases of subsets $X \subseteq X_1$ considered above,
it is straightforward to check that,
if $S\subseteq X_1$ is a subset with $\sup(S)\in V_f$,
then $V_f\cap S \neq \emptyset$ (actually, as $\bot\not\in V_f$, the first case,
$S \subseteq \{\bot\}$, can be ignored).
\item
For the other claim, consider some Scott open subset $U\subseteq X_1$
with $U \neq \emptyset$ and $U \neq X_1$.
Then $\bot\not\in U$ and $\top\in U$.
The set $S:=\IN\times \{\omega\}$ is directed and satisfies 
$\sup(S)=\top \in U$. Hence, there exists some $n\in\IN$ with
$(n,\omega)\in U$. Let $n$ be the smallest number with this property.
Note that $U \cap (\IN\times \{\omega\}) = \IN_n \times\{\omega\}$
as $U$ is upwards closed.
Now, for each $i\geq n$ the set 
$S_i:=\{i\}\times \IN$ is directed and satisfies
$\sup(S_i) = (i,\omega) \in U$.
Hence, there exists some $m_i\in\IN$ with
$(i,m_i)\in U$. We define $f:\IN_n\to\IN$ by letting $f(i)$ be the smallest
number $m_i$ with $(i,m_i)\in U$.
As $\bot\not\in U$ and $U$ is upwards closed, we obtain $U=V_f$.
\qedhere
\end{enumerate}
\end{proof}

Let $X_2$ be the set of all Scott open subsets of  $X_1$,
and let the binary relation $\sqsubseteq_2$ on $X_2$ be set-theoretic inclusion.
Then $(X_2,\sqsubseteq_2)$ is a complete lattice \cite[Examples O-2.7(3)]{GHKLMS03}.
Let 
\[ Z := X_1 \times X_2 , \]
and let the binary relation $\sqsubseteq_\times$ on $Z$ be defined as at the beginning
of this section.
Then, according to Lemma~\ref{lemma:prod}.3, $(Z,\sqsubseteq_\times)$
is a complete lattice as well.

From now on, when we speak about {\em the product topology} on a product
of two posets, we will always mean the product of the Scott topologies
of the two posets.
Remember that by Lemma~\ref{lemma:elementof}
the set $E \subseteq Z$ defined by
\[ E := \{(x_1,x_2) \in X_1 \times X_2 ~:~ x_1 \in x_2\} \]
is open in the Scott topology on $(Z, \sqsubseteq_\times)$.

\begin{lemma}
The subset $E \subseteq Z$
is not open with respect to the product topology on $Z = X_1 \times X_2$.
\end{lemma}

\begin{proof}
First, let us choose an element of $E$ as follows.
Let $\tilde{0}:\IN\to\IN$ be the constant number function with $\tilde{0}(n)=0$
for all $n\in\IN$. Then, according to Lemma~\ref{lemma:mod-johnstone-1},
the set $x_2:=  V_{\tilde{0}} = X_1 \setminus\{\bot\}$ is 
a Scott open subset of $X_1$. The element $x_1 := (0,0) \in X_1$ is an element of $x_2$.
Thus, $(x_1,x_2) \in E$.

Now, for the sake of a contradiction, let us assume that $E$ is open in the product
topology. Then there exist Scott open subsets $U\subseteq X_1$ and
$\mathcal{V} \subseteq X_2$ with $(x_1,x_2) \in U \times \mathcal{V}$
and $U \times \mathcal{V} \subseteq E$.
According to Lemma~\ref{lemma:mod-johnstone-1} there exist some $n\in\IN$
and some function $f:\IN_n\to \IN$ such that $U=V_f$.
As $(0,0)=x_1 \in U$, the number $n$ must be equal to $0$. 
For $i\in\IN$ we define
$g_i:\IN\to\IN$ by
\[ g_i(m) := \begin{cases}
       0 & \text{ if } m < i, \\
       f(m)+1 & \text{ if } m \geq i .
       \end{cases} \]
Then the sequence $(V_{g_i})_{i\in\IN}$ is an increasing sequence
of Scott open subsets of $X_1$
with 
\[ \sup( \{V_{g_i} ~:~ i \in\IN\}) = \bigcup_{i\in\IN} V_{g_i}
    = V_{\tilde{0}} = X_1 \setminus\{\bot\} .\]
Hence, there must exist some $i\in\IN$ with $V_{g_i} \in \mathcal{V}$.
As $U \times \mathcal{V} \subseteq E$, we obtain
$(x,V_{g_i}) \in E$, for all $x \in U$, hence $x \in V_{g_i}$, for all $x \in U$.
That means, $U \subseteq V_{g_i}$.
But this is false!
\end{proof}

\begin{proof}[Proof of Proposition~\ref{proposition:mod-johnstone}.]
Consider the complete lattices $(X_1,\sqsubseteq_1)$ and $(X_2,\sqsubseteq_2)$
defined above.
According to Lemma~\ref{lemma:prod}.4 the Scott topology on $X_1 \times X_2$
is at least as fine as the product topology on $X_1 \times X_2$.
And we have seen that there exists a subset 
$E \subseteq X_1 \times X_2$ that is open in the Scott topology on $X_1 \times X_2$ but
not open in the product topology on $X_1 \times X_2$.
\end{proof}

\section{Discontinuity of the Supremum Function with respect to the
Product Topology on a Complete Lattice}
\label{section:sup}

Let $(X,\sqsubseteq)$ be a sup semilattice.
Let $\supt:X\times X \to X$ be defined by
\[ \supt(x,y) := \sup\{x,y\} . \]
The following observation is elementary and well known.
We omit its proof.

\begin{proposition}
\label{proposition:sup-scott-cont}
Let $(X,\sqsubseteq)$ be a sup semilattice.
Then the function $\sup^{(2)}:X \times X\to X$ is continuous with respect
to the Scott topology on $X\times X$ (on the left hand side)
and the Scott topology on $X$ (on the right hand side).
\end{proposition}

In view of this observation the natural question arises whether
the function $\sup^{(2)}:X \times X\to X$ is still continuous if on $X\times X$
one does not consider the Scott topology but the product topology of the 
Scott topologies on each copy of $X$. 
This question has also been discussed in \cite{GHKLMS03}; see the
quotation in the introduction.
But it seems that this question has not been answered so far in the general case;
see the discussion in the introduction.
It is the main goal of this section to show that the answer to this question is in general no.

\begin{theorem}
\label{theorem:main}
There exists a complete lattice $(Z,\sqsubseteq)$ such that the binary supremum function
$\supt:Z \times Z \to Z$ is not continuous with respect 
to the product topology on $Z\times Z$ of the Scott topology on each copy of $Z$
(on the left hand side) and the Scott topology on $Z$ (on the right hand side).
\end{theorem}

Note that due to Proposition~\ref{proposition:sup-scott-cont}
and Lemma~\ref{lemma:prod}.(4)
this implies that the Scott topology on $Z \times Z$ is strictly finer
than the product topology on $Z \times Z$.
Thus, this theorem improves Proposition~\ref{prop:johnstone}
and Proposition~\ref{proposition:mod-johnstone}.

\begin{proof}
According to Proposition~\ref{proposition:mod-johnstone} there exist
complete lattices $(X_1,\sqsubseteq_1)$ and $(X_2,\sqsubseteq_2)$
such that the Scott topology on $X_1\times X_2$ is strictly finer than
the product topology on $X_1\times X_2$.
We choose such $X_1$ and $X_2$ and define $Z:=X_1\times X_2$.
Note that due to Lemma~\ref{lemma:prod}, $(Z,\sqsubseteq_\times)$ is a
complete lattice as well.
Let $E\subseteq Z = X_1\times X_2$ be a subset
that is open in the Scott topology on $X_1\times X_2$
but not open in the product topology on $X_1\times X_2$.
We are going to show that the preimage
\[ D:= (\supt)^{-1}(E) = \{ (z,z') \in Z \times Z ~:~ \supt(z,z') \in E\} \]
of $E$ under the function $\supt$ is not open in the product topology on $Z\times Z$
(of the Scott topologies on each copy of $Z$).
For the sake of a contradiction, let us assume that $D$
is open in the product topology on $Z\times Z$.
We are going to show that this assumption would imply that $E$ is open in 
the product topology on $Z = X_1 \times X_2$.

Let us fix some element $(x_1,x_2) \in E$.
It is sufficient to show that the assumption that $D$
is open in the product topology on $Z\times Z$ implies that
there exist a Scott open subset $E_1 \subseteq X_1$
and a Scott open subset $E_2 \subseteq X_2$ with $x_1 \in E_1$, with $x_2 \in E_2$
and with $E_1 \times E_2 \subseteq E$.
How do we arrive at such sets?
Let us write $\bot_i := \inf(X_i)$, for $i=1,2$.
Note that the element $(x_1,x_2) \in E \subseteq Z = X_1\times X_2$
is the supremum of the elements $(x_1,\bot_2)$ and $(\bot_1,x_2)$. Thus,
\[ ( (x_1,\bot_2), (\bot_1,x_2) ) \in D . \]
Our assumption that $D$
is open in the product topology on $Z\times Z$ implies
that there exist two Scott open subsets $D_1,D_2 \subseteq Z$ with 
$(x_1,\bot_2) \in D_1$, with $(\bot_1,x_2) \in D_2$, and with 
$D_1 \times D_2 \subseteq D$.
We define
\begin{eqnarray*}
  E_1 &:=& \{x \in X_1 ~:~ (x,\bot_2) \in D_1 \}, \\
  E_2 &:=& \{y \in X_2 ~:~ (\bot_1,y) \in D_2 \}.
\end{eqnarray*}
It is clear that $x_1 \in E_1$. As $D_1$ is upwards closed, the set $E_1$
is upwards closed as well. If $S\subseteq X_1$ is a directed set with 
$\sup(S) \in E_1$ then the set $\{(x,\bot_2) ~:~ x \in S\} \subseteq Z$
is a directed set as well, and its supremum exists and is an element of $D_1$.
As $D_1$ is Scott open we conclude that there exists some
$x\in S$ with $(x,\bot_2) \in D_1$, hence, with $x\in E_1$.
This shows that $E_1 \subseteq X_1$ is a Scott open set with $x_1 \in E_1$.
In the same way one shows that $E_2 \subseteq X_2$
is a Scott open set with $x_2 \in E_2$.
Finally, we claim $E_1 \times E_2 \subseteq E$.
Consider some $x \in E_1$ and $y \in E_2$.
Then $(x,\bot_2) \in D_1$ and $(\bot_1,y) \in D_2$, hence,
\[ ( (x,\bot_2), (\bot_1,y) ) \in D_1 \times D_2 \subseteq D . \]
This implies 
\[ \supt  ( (x,\bot_2), (\bot_1,y) ) \in E . \]
On the other hand we calculate
\[ \supt  ( (x,\bot_2), (\bot_1,y) )
  = (\supt(x,\bot_1),\supt(\bot_2,y))
  = (x,y) . \]
Thus, we obtain $(x,y) \in E$. This shows $E_1 \times E_2 \subseteq E$ and
ends the proof of Theorem~\ref{theorem:main}.
\end{proof}

\section{Failure of Bounded Completeness for the
Directed Complete Partial Order of Scott Continuous Functions}
\label{section:bc}

Consider now some poset $(Z,\sqsubseteq)$ and an arbitrary topological space $X$.
We call a function $f:X\to Z$ {\em Scott continuous} if it is continuous with respect
to the given topology on $X$ and the Scott topology on $Z$.
Let $\mathrm{C}(X,Z)$ denote the set of all Scott continuous functions $f:X\to Z$.
On this set we define a binary relation $\sqsubseteq_\mathrm{C}$ by
\[ f \sqsubseteq_\mathrm{C} g :\iff (\forall x \in X)\ f(x) \sqsubseteq g(x) .\]

\begin{proposition}[{\cite[Lemma 1-4.6]{KL14}, \cite[Prop. 6]{Her16a}}]
\label{proposition:function-dcpo}
Let $(Z,\sqsubseteq)$ be a dcpo, 
and let $X$ be an arbitrary topological space.
Then $\mathrm{C}(X,Z)$ with $\sqsubseteq_\mathrm{C}$ is a dcpo.
Furthermore, if $F\In \mathrm{C}(X,Z)$ is a $\sqsubseteq_\mathrm{C}$-directed set 
then the function $g:X\to Z$ defined by 
\[ g(x):=\sup( \{f(x) ~:~ f\in F\}) \]
is Scott continuous and the least upper bound of $F$.
\end{proposition}

Let $(Z,\sqsubseteq)$ be a poset.
\begin{itemize}
\item
A subset $S \In Z$ is called {\em bounded} if there exists an upper bound $z\in Z$ for $S$.
\item
The poset $(Z,\sqsubseteq)$ is called {\em bounded complete}
if for any bounded subset $S \In Z$ there exists a supremum of $S$ in $Z$.
\item
An element of $Z$ is called a {\em least element} of $Z$ if it is a lower
bound of $Z$. Obviously, if a least element exists then it is unique.
\end{itemize}

Note that any bounded complete poset $(Z,\sqsubseteq)$
has a  least element
(the element $\sup(\emptyset)$ does the job).
And note that any complete lattice is a dcpo and that
a poset is a complete lattice if, and only if, it is bounded and bounded complete.

\begin{corollary}
\label{cor:bounded}
Let $X$ be an arbitrary topological space.
\begin{enumerate}
\item
If $(Z,\sqsubseteq)$ is a bounded dcpo then
$(C(X,Z),\sqsubseteq_\mathrm{C})$ is a bounded dcpo as well.
\item
If $(Z,\sqsubseteq)$ is a dcpo with least element then
$(C(X,Z),\sqsubseteq_\mathrm{C})$ is a dcpo with least element as well.
\end{enumerate}
\end{corollary}

\begin{proof}
Let $(Z,\sqsubseteq)$ be a dcpo.
By Prop.~\ref{proposition:function-dcpo},
$(C(X,Z),\sqsubseteq_\mathrm{C})$ is a dcpo.
\begin{enumerate}
\item
Let $\top \in Z$ be an upper bound for $Z$.
The constant function $f:X\to Z$ with $f(x):=\top$ for all $x\in X$
is Scott continuous, thus, an element of $C(X,Z)$, and it is an upper bound
for $C(X,Z)$.
\item
Let $\bot \in Z$ be a lower bound for $Z$.
The constant function $f:X\to Z$ with $f(x):=\bot$ for all $x\in X$
is Scott continuous, thus, an element of $C(X,Z)$, and it is a lower bound
for $C(X,Z)$.
\qedhere
\end{enumerate}
\end{proof}

The following question arises:
If $(Z,\sqsubseteq)$ is bounded complete, is then $C(X,Z)$ bounded complete as well?
This is true if the topology considered
on $X$ is the Scott topology induced by a partial order relation on $X$.

\begin{proposition}
\label{prop:bounded-complete}
Let $(X,\sqsubseteq_X)$ be a poset and consider on $X$ the Scott topology.
If $(Z,\sqsubseteq)$ is a bounded complete dcpo then
$(C(X,Z),\sqsubseteq_\mathrm{C})$ is a bounded complete dcpo as well.
\end{proposition}

This proposition is most easily shown using the following fundamental
characterization of Scott continuous functions between partial orders.

\begin{lemma}[{see, e.g., \cite[Prop.~4.3.5]{Gou13}}]
\label{lemma:scott-continuous}
Let $(X,\sqsubseteq_X)$ and $(Z,\sqsubseteq_Z)$ be partial orders.
For a function $f:X\to Z$ the following two conditions are equivalent.
\begin{enumerate}
\item
$f$ is Scott continuous, that is, continuous with respect to the Scott topology on $X$
and the Scott topology on $Z$.
\item
$f$ is monotone (that is, $(\forall x,x' \in X) (x \sqsubseteq_X x'
\Rightarrow f(x) \sqsubseteq_Z f(x')$) and, if $S\subseteq X$
is a directed set whose supremum $\sup(X)$ exists, then $\sup(f(X))$
exists as well and satisfies $\sup(f(X))=f(\sup(X))$.
\end{enumerate}
\end{lemma}

\begin{proof}[Proof of Prop.~\ref{prop:bounded-complete}]
Let $F\In C(X,Z)$ be a $\sqsubseteq_\mathrm{C}$-bounded set.
Then for each $x\in X$ the set
\[ F(x):=  \{f(x) ~:~ f \in F\}
  = \{ z \in Z ~:~ (\exists f \in F) \ z=f(x) \} 
\]
is bounded. Since we assume $(Z,\sqsubseteq_Z)$ to be bounded complete, we can
define a function $g:X\to Z$ by
\[ g(x) := \sup(F(x)) . \]

First, we show that this function $g$ is Scott continuous. We use the characterization
in Lemma~\ref{lemma:scott-continuous}.
First we show that $g$ is monotone.
Let us consider $x,y \in X$ with $x\sqsubseteq_X y$.
As each function $f\in F$ is monotone, we have for each function $f\in F$
\[ f(x) \sqsubseteq_Z f(y) \sqsubseteq_Z g(y), \]
hence, $g(x) \sqsubseteq_Z g(y)$. Thus, $g$ is monotone.
Now, let $S\In X$ be a directed set. Then the set $g(S)$ is directed as well 
because $g$ is monotone.
On the one hand, for every $s\in S$ we have
$g(s) \sqsubseteq_Z g(\sup(S))$ because $g$ is monotone,
hence, 
\[ \sup(g(S)) \sqsubseteq_Z g(\sup(S)) . \]
On the other hand, 
for every $f\in F$ and every $s\in S$,
\[ f(s) \sqsubseteq_Z g(s) \sqsubseteq_Z \sup(g(S)) . \]
Thus, for every $f\in F$,
\[ \sup(f(S)) \sqsubseteq_Z \sup(g(S)) . \]
Finally, every $f\in F$ is Scott continuous.
Therefore, for each $f\in F$ we have
\[ f( \sup(S) ) = \sup(f(S)) . \]
The last two formulae together give
\[ f( \sup(S)) \sqsubseteq_Z \sup(g(S)) \]
for each $f\in F$. We obtain
\[ g( \sup(S) ) = \sup \{ f(\sup(S)) ~:~ f \in F\}
   \sqsubseteq_Z \sup(g(S)) . \]
We have shown $g(\sup(S)) = \sup(g(S))$.
This completes the proof of our claim that
$g$ is Scott continuous.

It is obvious that $g$ is an upper bound of $F$.
If $h \in C(X,Z)$ is any upper bound of $F$ then
for all $x\in X$ and all $f\in F$,
$f(x) \sqsubseteq_Z h(x)$, hence,
$g(x) = \sup \{f(x) ~:~ f \in F\} \sqsubseteq_Z h(x)$, hence,
$g \sqsubseteq_\mathrm{C} h$.
We have shown that $g$ is the least upper bound of $F$.
\end{proof}

One may now speculate whether for an arbitrary topological space $X$
and an arbitrary bounded complete dcpo $Z$ the dcpo
$C(X,Z)$ is bounded complete as well; compare
the discussion in the introduction.
But it is the second main goal of this note to show that in general this is not the case.

\begin{theorem}
\label{theorem:bounded-complete-counterexample}
There exist a topological space $X$ and a complete lattice $Z$ 
with the following properties:
\begin{enumerate}
\item
$C(X,Z)$ is a bounded dcpo with least element, but not bounded complete.
\item
There exist two functions $f,g\in C(X,Z)$ with the following two properties.
\begin{enumerate}
\item
The set $\{f,g\}$ does not have a supremum in $C(X,Z)$.
\item
The {\em pointwise supremum} of $f$ and $g$, that is, the function
$h:X \to Z$ defined by $h(x) := \sup\{f(x),g(x)\}$, is not Scott continuous.
\end{enumerate}
\end{enumerate}
\end{theorem}

It is clear that Property (a) in Theorem~\ref{theorem:bounded-complete-counterexample}
implies Property (b). 
The converse in not so obvious, but true as well, as the following lemma shows.
It will be used in the proof of Theorem~\ref{theorem:bounded-complete-counterexample}.

\begin{lemma}
\label{lemma:KL}
For any topological space $X$, any complete lattice $Z$ and any two
functions $f,g \in C(X,Z)$ the following two conditions are equivalent.
\begin{enumerate}
\item
The set $\{f,g\}$ has a supremum in $C(X,Z)$.
\item
The pointwise supremum $h:X\to Z$ of $f$ and $g$ defined by
\[ h(x) := \sup\{f(x),g(x)\} \]
for all $x\in X$ is Scott continuous.
\end{enumerate}
Furthermore, if one (and then both) of these two conditions is satisfied
then the pointwise supremum of $f$ and $g$ is the supremum of $\{f,g\}$ in $C(X,Z)$.
\end{lemma}

\begin{proof}
Fix some functions $f,g \in C(X,Z)$, and let $h:X\to Z$ be the pointwise
supremum of $f$ and $g$.

``$2\Rightarrow 1$'':
If $h$ is continuous and $r \in C(X,Z)$ any function with $f \sqsubseteq_\mathrm{C} r$
and $g \sqsubseteq_\mathrm{C} r$ then $h \sqsubseteq_\mathrm{C} r$.
Thus, in this case $h$ is a supremum of $\{f,g\}$ in $C(X,Z)$.

For the converse direction, let us first show the following claim.

{\bf Claim.}
For any $x_0 \in X$ there exists a function $s_{x_0} \in C(X,Z)$ with
$f \sqsubseteq_\mathrm{C} s_{x_0}$, with
$g \sqsubseteq_\mathrm{C} s_{x_0}$ and with
$s_{x_0}(x_0) = \sup\{f(x_0),g(x_0)\}$.

{\em Proof of this claim.}
Let $\top:=\sup(Z)$. 
Let us fix an arbitrary point $x_0\in X$, and let $\cl(\{x_0\})$ be the closure
of the set $\{x_0\}$, that is,
\[ \cl(\{x_0\}) = \{x \in X ~:~ \text{ for every open } U \subseteq X,
   \text{ if } x \in U \text{ then } x_0 \in U\} .\]
Let us define a ``step function'' $s_{x_0}:X\to Z$ by
\[ s_{x_0}(x) :=    \begin{cases} 
       \sup\{f(x_0),g(x_0)\} & \text{ if } x \in \cl(\{x_0\}), \\
       \top & \text{ otherwise.}
       \end{cases} \]
We claim that this function has the desired properties.
First, we claim that $s_{x_0}$ is Scott continuous.
This is clear if $\sup\{f(x_0),g(x_0)\} = \top$. Also in the case
$\sup\{f(x_0),g(x_0)\} \neq \top$ for any Scott open subset $V\subseteq Z$ the set
\[ s_{z_0}^{-1}(V) =        
          \begin{cases} 
             \emptyset & \text{ if } V = \emptyset, \\
             X \setminus \cl(\{x_0\})
              & \text{ if } V\neq \emptyset \text{ and } \sup\{f(x_0),g(x_0)\} \not\in V, \\
              X & \text{ if } \sup\{f(x_0),g(x_0)\} \in V
\end{cases} \]
is an open subset of $X$. Hence, the function $s_{z_0}$ is continuous,
thus, an element of $C(X,Z)$.
We claim that furthermore $f \sqsubseteq_\mathrm{C} s_{x_0}$.
Indeed, if $x\not\in \cl(\{x_0\})$ then $f(x) \sqsubseteq_Z \top = s_{x_0}(x)$.
Let us now consider the case $x\in \cl(\{x_0\})$.
The set
\[ \downarrow f(x_0) := \{ z\in Z ~:~ z \sqsubseteq f(x_0) \}\]
is a Scott closed subset of $Z$, that is, its complement
$Z \setminus \downarrow f(x_0)$ is a Scott open subset of $Z$.
As $f:X\to Z$ is Scott continuous, the set
$U:=f^{-1}(Z \setminus \downarrow f(x_0))$
is an open subset of $X$.
If $x$ were an element of $U$ then $x_0$ would have to be an element of $U$ as well,
which is false.
Thus, $x\not\in U$. This means $f(x) \sqsubseteq_Z f(x_0)$. We obtain
\[ f(x) \sqsubseteq_Z f(x_0) \sqsubseteq_Z \sup\{f(x_0,g(x_0)\} = s_{x_0}(x) . \]
We have shown $f \sqsubseteq_\mathrm{C} s_{x_0}$. By the same argument 
$g \sqsubseteq_\mathrm{C} s_{x_0}$ follows. 
Finally, it is clear that $s_{x_0}(x_0) = \sup\{f(x_0),f(x_0)\}$.
We have shown the claim.

``$1\Rightarrow 2$'':
Let us assume that the set $\{f,g\}$ has a supremum in $C(X,Z)$.
Let us call this supremum $t$.
Let us consider some $x_0 \in X$.
Clearly $h(x_0) = \sup\{f(x_0),g(x_0)\} \sqsubseteq_Z t(x_0)$.
According to the claim that we have just proved, for any $x_0$ there exists a function
$s_{x_0} \in C(X,Z)$ with
$f \sqsubseteq_\mathrm{C} s_{x_0}$, with
$g \sqsubseteq_\mathrm{C} s_{x_0}$ and with
$s_{x_0}(x_0) = \sup\{f(x_0),g(x_0)\}$.
This implies $t = \sup\{f,g\} \sqsubseteq_\mathrm{C} s_{x_0}$ and in particular
$t(x_0) \sqsubseteq_Z s_{x_0}(x_0) = \sup\{f(x_0),g(x_0)\} = h(x_0)$.
Thus, we do not only have $h(x_0) \sqsubseteq_Z t(x_0)$ but also
$t(x_0) \sqsubseteq_Z h(x_0)$, thus, $t(x_0)=h(x_0)$.
As this is true for all $x_0\in X$, we have $t=h$,
that is, $t$ must be the pointwise supremum of $f$ and $g$.
\end{proof}

\begin{proof}[Proof of Theorem~\ref{theorem:bounded-complete-counterexample}]
Let $(Z,\sqsubseteq)$ be a complete lattice as in Theorem~\ref{theorem:main}. Then
$Z$ is bounded ($\sup(Z)$ is an upper bound) as well as
bounded complete and directed complete (every subset has a supremum),
and it has a least element.
For $X$ we choose $Z\times Z$ with the product topology
of the Scott topologies on each copy of $Z$.
By Cor.~\ref{cor:bounded} $C(X,Z)$ is a bounded dcpo with least element.
Once we have shown that there exist two functions $f,g\in C(X,Z)$
that do not have a supremum in $C(X,Z)$
we have shown that $C(X,Z)$ is not bounded complete.

In fact, for $f$ and $g$ we can take the projection functions. For $i=1,2$
we define $\pi_i:X=Z\times Z \to Z$ by
$\pi_i(z_1,z_2):=z_i$. Then $\pi_i$ is continuous with respect to the product
topology on $X=Z\times Z$, thus, $\pi_i \in C(X,Z)$.
Note that the function $\supt$ is the pointwise supremum of $\pi_1$ and $\pi_2$.
As we have chosen $Z$ to be a complete lattice
as in Theorem~\ref{theorem:main}
and $X$ to be $Z\times Z$ with the product topology,
the function $\supt$ is not an element of $C(X,Z)$, that is,
$\supt$ is not Scott continuous. This shows the final claim of the theorem.
And according to Lemma~\ref{lemma:KL}, the set $\{\pi_1,\pi_2\}$
does not have a suprumum in $C(X,Z)$.
This shows the other claims of the theorem.
This ends the proof of Theorem~\ref{theorem:bounded-complete-counterexample}.
\end{proof}

\end{document}